\documentclass[a4paper,english, 11pt]{article}
\usepackage[utf8]{inputenc}
\usepackage{mathtools}
\usepackage{amsmath}
\usepackage{amssymb}
\usepackage{stmaryrd}
\usepackage{amsfonts}
\usepackage{amsthm}
\usepackage{thmtools}
\usepackage{thm-restate}
\usepackage[english]{babel}
\usepackage[hidelinks]{hyperref}
\usepackage{cleveref}
\usepackage{enumitem}
\usepackage{geometry}
\usepackage{graphicx}
\usepackage{listings}
\usepackage{multicol}
\usepackage{float}
\usepackage{subfiles}
\usepackage{comment}
\usepackage{pdfpages}
\usepackage{appendix}
\DeclareMathAlphabet{\mathbbm}{U}{bbold}{m}{n}
\usepackage{a4wide}
\usepackage{mathrsfs}
\usepackage[nottoc]{tocbibind}
\usepackage{ifthen}
\usepackage{tikz-cd}
\usepackage{color,soul}
\usepackage{framed}

\newcommand{\grad}{\operatorname{grad}}

\newcommand{\N}{\mathbb{N}}

\newcommand{\R}{\mathbb{R}}

\renewcommand{\tilde}{\widetilde}

\usepackage[most]{tcolorbox}
\newcommand{\Sph}{\mathbb{S}}

\newcommand{\Int}{\operatorname{Int}}

\let\originalleft\left
\let\originalright\right
\renewcommand{\left}{\mathopen{}\mathclose\bgroup\originalleft}
\renewcommand{\right}{\aftergroup\egroup\originalright}

\newlist{thmlist}{enumerate}{1}
\setlist[thmlist]{label=(\roman{thmlisti}), ref=\thetheorem(\roman{thmlisti}), noitemsep}
\Crefname{theorem}{Theorem}{Theorems}
\Crefname{listthm}{Theorem}{Theorems}
\addtotheorempostheadhook[theorem]{\crefalias{thmlisti}{listthm}}

\addto\extrasenglish{%
}

\newtheorem{theorem}{Theorem}[section]
\newtheorem{lemma}[theorem]{Lemma}
\newtheorem{proposition}[theorem]{Proposition}
\newtheorem{corollary}[theorem]{Corollary}
\theoremstyle{definition}
\newtheorem{definition}[theorem]{Definition}
\newlist{deflist}{enumerate}{1}
\setlist[deflist]{label=(\roman{deflisti}), ref=\thedefinition(\roman{deflisti}), noitemsep}
\Crefname{definition}{Definition}{Definitions}
\Crefname{listdef}{Definition}{Definitions}
\addtotheorempostheadhook[definition]{\crefalias{deflisti}{listdef}}
\newtheorem*{example}{Example}
\newtheorem*{remark}{Remark}
\counterwithout{footnote}{section}
\allowdisplaybreaks

\newcommand{\symfootnote}[1]{%
\let\oldthefootnote=\thefootnote%
\stepcounter{mpfootnote}%
\addtocounter{footnote}{-1}%
\renewcommand{\thefootnote}{\fnsymbol{mpfootnote}}%
\footnote{#1}%
\let\thefootnote=\oldthefootnote%
}

\newcommand\blfootnote[1]{%
  \begingroup
  \renewcommand\thefootnote{}\footnote{#1}%
  \addtocounter{footnote}{-1}%
  \endgroup
}

\begin{document}

{\blfootnote{\textit{2020 Mathematics Subject Classification.} 83C30 (primary), 53C18, 53C50, 58K55, 83C57}}
{\blfootnote{\textit{Keywords and phrases.} Multipole moments, asymptotic flatness, stationarity, spacetimes, general relativity, conformal completions, spatial infinity.}}
\begin{center}
	{\huge \bf Multipole moments in stationary spacetimes}

    \vspace{\baselineskip}

    {\large Jorn van Voorthuizen\symfootnote{Mathematisches Institut, Universität zu Köln, Weyertal 86-90, D-50931 Köln, Germany. \\ Email: \href{mailto:jvoorthu@uni-koeln.de}{jvoorthu@uni-koeln.de}.}}
\end{center}

\vspace{2\baselineskip}

\begin{abstract}
    Multipole moments in general relativity serve as a powerful tool for characterising the gravitational field. In this paper, we review the construction of the Geroch--Hansen multipole moments for stationary asymptotically flat vacuum spacetimes. A particular focus is placed on the well-definedness of these moments, which hinges on the uniqueness of the one-point conformal completion in Geroch's asymptotic flatness definition. Based on Geroch's approach, we formulate and prove a revised uniqueness result, thereby filling in some gaps in the original approach. Uniqueness holds up to certain conformal transformations, and we discuss how the multipole moments behave under such transformations.
\end{abstract}

\section{Introduction}

Multipole moments encode the angular dependence of certain fields (e.g., electromagnetic or gravitational fields). In stationary asymptotically flat vacuum spacetimes, the most important definitions for multipole moments are due to Geroch \cite{gerochMultipoleMomentsII1970} and Hansen \cite{hansenMultipoleMomentsStationary1974} in the 1970s and due to Thorne \cite{thorneMultipoleExpansionsGravitational1980} in 1980. The Geroch--Hansen formalism is constructed using a conformal completion and is inherently coordinate-invariant. On the other hand, Thorne's multipole moments are based on finding a suitable coordinate system. Surprisingly, both approaches are equivalent, as shown by Gürsel \cite{gurselMultipoleMomentsStationary1983} in 1983.

As already mentioned in the title, we restrict our attention to the class of stationary spacetimes. There exist definitions of multipole moments for spacetimes with arbitrary time dependence, but they often describe only linear perturbations rather than the exact gravitational field \cite{janisStructureGravitationalSources1965,lambMultipoleMomentsEinstein1966,thorneMultipoleExpansionsGravitational1980,willmerMultipoleMomentsGeneral1981}. An exception is the multipole moments due to Compère, Oliveri and Seraj \cite{Compere2018gravitational}, who used a coordinate approach in the same spirit as Thorne and some Noether charges. In stationary asymptotically flat spacetimes, the Geroch--Hansen and Thorne multipole moments also describe the full nonlinear picture. In their work, it was also assumed that the spacetime is a vacuum solution (without cosmological constant), but the formalisms have been generalised to other classes of solutions (e.g., electrovacuum \cite{simonMultipoleExpansionStationary1984}) and, recently, to spacetimes with generic matter fields \cite{mayersonGravitationalMultipolesGeneral2023}.

The main reason why the Geroch--Hansen and Thorne multipole moments are interesting, is that they characterise a spacetime. This was independently shown by Beig and Simon \cite{beigMultipoleExpansionStationary1981} and Kundu \cite{kunduAnalyticityStationaryGravitational1981}, both in 1981. That makes it meaningful to measure multipole moments and compare Einstein's theory of general relativity to alternative theories of gravity, provide characteristics for gravitational waves, and test results/conjectures in general relativity \cite{babakScienceSpacebasedInterferometer2017,barackUsingLISAExtrememassratio2007,cardosoTestingBlackHole2016,johannsenTestingNoHairTheorem2010,liGeneralizationRyanTheorem2008,ryanGravitationalWavesInspiral1995,ryanAccuracyEstimatingMultipole1997}.

In this paper, we review and discuss the Geroch--Hansen formalism for stationary asymptotically flat vacuum spacetimes. Along the way, we improve two existing results to fill in the gaps in the theory of multipole moments. These new results are \autoref{thm:uniqueness} and \autoref{cor:mms_conftransi0}. Next, we briefly discuss these results and the overall structure of this paper.

To define multipole moments, we use asymptotic flatness in the sense defined by Geroch \cite{gerochMultipoleMomentsII1970} (see \autoref{def:asympflatGeroch}). The idea is to perform a conformal completion on a three-dimensional Riemannian manifold by adding a single point $i^0$ at infinity. In stationary spacetimes, one can think of $i^0$ as spatial infinity. The multipole moments $P^k$, $k\in \N_{0}=\{0,1,2,\dots\}$, are tensors at this added point $i^0$. They are defined by a recursion relation of consecutive derivatives, much like how we compute the coefficients of a Taylor expansion of an analytic field. The definition of multipole moments can be found in \autoref{def:seqtensors}. In order for the multipole moments to be well-defined, the conformal completion needs to be uniquely determined by the spacetime. As claimed by Geroch, there is the freedom of a conformal transformation that acts as an isometry at $i^0$. However, the proof in \cite{gerochMultipoleMomentsII1970} is incorrect. The topology constructed in Geroch's proof is not a topology. In our main result, we formulate and prove a correct uniqueness result.

\begin{restatable*}{theorem}{uniqueness}\label{thm:uniqueness}
    Let $(S,h)$ be a three-dimensional Riemannian manifold and let $K\subseteq S$ be a closed subset. Suppose there is a homeomorphism $\varphi\colon S\setminus \Int{K}\to \R^3\setminus \mathbb{B}^3$ which restricts to a diffeomorphism between $S\setminus K$ and $\R^3\setminus \overline{\mathbb{B}}^3$. Suppose $(S,h)$ is asymptotically flat and $\left(\tilde{S},\tilde{h}\right)$ is a conformal completion with a conformal factor $\Omega\in C^2\left(\tilde{S}\right)$ as in \autoref{def:asympflatGeroch} such that $\tilde{S}\setminus \Int{K}$ is compact, then $\left(\tilde{S},\tilde{h}\right)$ is unique up to conformal transformations with conformal factor $1$ at $i^0$.
\end{restatable*}

Here, $\mathbb{B}^3\subseteq \R^3$ denotes the open unit ball. We briefly discuss this theorem. We should think of $K\subseteq S$ as a bounded subset, and then it is not unreasonable to expect we can take $K$ such that $S\setminus K$ is diffeomorphic to $\R^3\setminus \overline{\mathbb{B}}^3$ as $S$ is asymptotically flat. The way we phrased it, via this homeomorphism $\varphi$, is mainly to distinguish the boundary of $S\setminus K$ from the unbounded region where we want to add the point $i^0$ at infinity. The added compactness assumption ensures that we have a one-point compactification of $\tilde{S}\setminus \Int{K}$, fixing the topology on this subset uniquely. This provides enough structure to ensure $i^0$ is added in a unique way up to some residual conformal transformations.

In \autoref{thm:uniqueness}, there is still freedom in the conformal completion, which is why we investigate how the multipole moments $P^k$ transform under these residual conformal transformations. The formula is given in \autoref{cor:mms_conftransi0}. This result was obtained by Beig \cite{beigMultipoleExpansionGeneral1981} for static spacetimes. There is no essential difference in our approach for stationary spacetimes, but we provide a detailed proof. \newline

\textbf{Outline.} In \autoref{sec:assumptions}, we discuss our assumptions on the spacetime: stationarity, asymptotic flatness, and a vacuum solution of the Einstein equations. Most notably, we discuss uniqueness of the conformal completion for three-dimensional manifolds in \autoref{sec:asympflat}. In \autoref{sec:mms}, we discuss the construction of multipole moments and, in particular, how they transform under the residual conformal transformations. We end with a brief discussion on the results and how they can be extended in \autoref{sec:discussion}. In \autoref{app:STF}, we recall some identities for symmetric trace-free tensors. \newline

\textbf{Notation.} We restrictively use $(M,g)$ to denote a spacetime, i.e., a connected four-dimensional Lorentzian manifold with a preferred time orientation. A stationary vector field is denoted by $\xi$, and once we fix a stationary vector field, we also assume it determines the time orientation. We denote a three-dimensional Riemannian manifold by $(S,h)$. The Levi-Civita connection on $(S,h)$ is denoted by $D$. The conformal completion of $(S,h)$ in the sense of \autoref{def:asympflatGeroch} (if it exists) is denoted by $\left(\tilde{S},\tilde{h}\right)$, and $\tilde{D}$ is the corresponding Levi-Civita connection. \newline

\textbf{Acknowledgements.} This work is based on the author's master's thesis at Radboud University. I want to thank my supervisors Béatrice Bonga and Annegret Burtscher for their supervision and their valuable comments on preliminary versions of this work. I also want to thank my PhD advisor Ioan M\u{a}rcu\textcommabelow{t} and the University of Cologne for their support in working on this paper.

\section{Geometric assumptions}\label{sec:assumptions}

The construction of the Geroch--Hansen multipole moments requires three assumptions on the spacetime. The spacetime must be stationary, asymptotically flat, and a solution of the Einstein equations in vacuum. The assumptions are not independent: asymptotic flatness depends on stationarity, and using the Einstein equations we define some potentials which must be asymptotically flat. In this section, we discuss the assumptions in this order.

\subsection{Stationarity}

Stationarity is used to reduce the setting for the construction of multipole moments from a $4$-dimensional spacetime to a $3$-dimensional Riemannian manifold. In this subsection, we discuss how to build this $3$-dimensional space.

\begin{definition}\label{def:stationary-complete}
    A spacetime $(M,g)$ is \emph{stationary} if there exists a complete timelike Killing vector field $\xi$ on $(M,g)$. I.e., $g(\xi,\xi)<0$ (timelike), $\mathcal{L}_\xi g=0$ (Killing) and the maximal integral curves of $\xi$ are defined on all of $\R$ (complete). We call $\xi$ a \emph{stationary vector field}.
\end{definition}

The most common definition of a stationary spacetime requires merely a timelike Killing vector field \cite{carrollSpacetimeGeometryIntroduction2019, harrisConformallyStationarySpacetimes1992, poissonRelativistToolkitMathematics2004}. I.e., the completeness assumption is dropped. In this case, the (maximal) flow of $\xi$ is a map $\theta\colon \mathcal{D}\to M$ with $\mathcal{D}\subseteq \R\times M$ such that $\theta^{(p)}(\cdot) = \theta(\cdot, p)$ is a (maximal) integral curve of $\xi$ starting at $p$. Let $M_t=\{p\in M : (t,p)\in \mathcal{D}\}$, then $M_t\subseteq M$ is an open subset and we can also endow $M_t$ with $g$. In that case, $\theta_t=M_t\to M_{-t}$ is an isometry because $\xi$ is a Killing vector field. It is easy to check that the flow $\theta$ induces an equivalence relation $\sim$ on $M$ given by $p\sim q$ if and only if there exists some $t\in \R$ such that $(t,p)\in \mathcal{D}$ and $q=\theta(t,p)$. The idea is to reduce the spacetime to a three-dimensional space by considering $M/\sim$. To this end, we also assume completeness of $\xi$. In that case $\mathcal{D}=\R\times M$ and $\theta$ is a smooth $\R$-action on $M$, simplifying quotients. The completeness assumption for a stationary vector field is sometimes also added in textbooks \cite{landsmanFoundationsGeneralRelativity2021, waldGeneralRelativity1984}.

If the stationary vector field is not complete, however, we may still take quotients by transforming $\xi$ into $f\xi$ for some smooth function $f$, such that $f\xi$ is complete. Then $f\xi$ is no longer a Killing vector field anymore. In that case, the quotient manifold is not ensured to be Hausdorff \cite{harrisConformallyStationarySpacetimes1992}. For the rest of this paper, $\xi$ is assumed to be complete.

\begin{definition}\label{def:obsspace}
    The \emph{observer space} of a stationary spacetime $(M,g,\xi)$ is the quotient space $S=M/\R$ under the flow of $\xi$.
\end{definition}

Using the quotient manifold theorem for quotients of Lie group actions, the observer space $S$ is a smooth manifold if the action is free and proper. To achieve this, we impose a causality condition. Harris \cite{harrisConformallyStationarySpacetimes1992} showed that it is sufficient to assume the weakest causality condition on the causal ladder \cite{minguzziLorentzianCausalityTheory2019}.

\begin{theorem}[{{\cite[Theorem 1]{harrisConformallyStationarySpacetimes1992}}}]
    Let $(M,g,\xi)$ be a non-totally vicious\footnote{A spacetime is non-totally vicious if there exists a point through which there is no closed timelike curve. In particular, chronological spacetimes are non-totally vicious. A stationary spacetime (with the completeness assumption) is non-totally vicious if and only if it is chronological \cite[Corollary 4.27]{minguzziLorentzianCausalityTheory2019}.}, stationary spacetime with observer space $S$, then there is a unique smooth structure on $S$ such that the projection $\pi\colon M\to S$ is a surjective smooth submersion.
\end{theorem}

We want to describe the spacetime in terms of the observer space. For example, we want to describe the Einstein equations as equations on $S$. To do so, we have to understand tensor fields on $S$.

\begin{proposition}[{{\cite[Appendix]{gerochMethodGeneratingSolutions1971}}}]
    Let $S$ be the observer space of a stationary spacetime $(M,g,\xi)$. There is a $C^\infty(S)$-module isomorphism between $\Gamma\left(T^{(k,l)}S\right)$ and the tensor fields $T\in\Gamma\left(T^{(k,l)}M\right)$ such that $\mathcal{L}_\xi T =0$ and all possible contractions between $T$ and $\xi$ vanish.
\end{proposition}

The first example of a tensor field on $S$ is the scalar field
\begin{equation}\label{eq:lambda}
    \lambda=-g(\xi,\xi).
\end{equation}
The observer space is a $3$-dimensional manifold, and can be equipped with a suitable Riemannian metric. Seen as a tensor on $M$, it is given by
\begin{equation}\label{eq:metric}
	h=\lambda g+\xi^\flat \otimes \xi^\flat,
\end{equation}
where $\xi^\flat$ is the one-form on $M$ given by $\xi^\flat(X)=g(\xi,X)$ for vector fields $X\in \mathfrak{X}(M)$. The metric \eqref{eq:metric} yields a convenient form for the Einstein equations expressed on $S$ \cite{gerochMethodGeneratingSolutions1971}.\footnote{In other situations, one considers the conformal metrics $\lambda^{-1}h$ and $\lambda^{-2}h$. Working with $\lambda^{-1}h$ can be convenient because it turns $\pi\colon M\to S$ into a pseudo-Riemannian submersion. If $(M,g)$ is globally hyperbolic, the metric $\lambda^{-2}h$ turns $S$ into a complete Riemannian manifold \cite[Theorem 8]{garfinkleRicciFalloffStatic1997}.} Thirdly, the \textit{twist one-form} is given by
\begin{equation}\label{eq:twistoneform}
    \omega = - *\left(\xi^\flat\wedge d\xi^\flat\right),
\end{equation}
and reduces to a one-form on $S$. The exterior derivative of $\omega$ is \cite[p. 164]{waldGeneralRelativity1984}
\begin{equation}\label{eq:exderomega}
    d\omega = -2*\left(\xi^\flat \wedge i_\xi Rc\right),
\end{equation}
where $Rc$ is the Ricci tensor for $(M,g)$ and $i_\xi Rc$ is the one-form given by $\left(i_\xi Rc\right)(X)=Rc(\xi,X)$ for $X\in \mathfrak{X}(M)$. In particular, the twist one-form is closed for vacuum solutions of the Einstein equations (without cosmological constant).

\subsection{Asymptotic flatness}\label{sec:asympflat}

Geroch's notion of asymptotic flatness of Riemannian $3$-manifolds is inspired by Penrose's definition for asymptotic flatness at null infinity \cite{penroseAsymptoticPropertiesFields1963}. Typically, asymptotic flatness is defined by some decay conditions in a certain chart. Chru\'{s}ciel \cite{chruscielStructureSpatialInfinity1989} discussed how the definition below compares to more coordinate-dependent definitions.

\begin{definition}[Geroch's asymptotic flatness, \cite{gerochMultipoleMomentsII1970}]\label{def:asympflatGeroch}
    A three-dimensional Riemannian manifold $(S,h)$ is called \textit{asymptotically flat} if there exists a Riemannian manifold $\left(\tilde{S}, \tilde{h}\right)$ and a function $\Omega\in C^2\left(\tilde{S}\right)$ such that:
    \begin{deflist}
        \item $\tilde{S}=S\cup\{i^0\}$ and the inclusion $\iota\colon S\to \tilde{S}$ is a smooth embedding;
        \item $\iota^*\tilde{h} = (\iota^*\Omega)^2 h$;
        \item $\Omega(i^0)=0$, $d\Omega_{i^0}=0$ and $\tilde{D}(d\Omega)\rvert_{i^0} = 2\tilde{h}_{i^0}$, where $\tilde{D}$ denotes the Levi-Civita connection of $\left(\tilde{S},\tilde{h}\right)$.\footnote{Geroch only assumes that $\tilde{D}\left(\tilde{D}\Omega\right)\rvert_{i^0}$ is proportional to $\tilde{h}_{i^0}$. The factor $2$ is due to Hansen \cite{hansenMultipoleMomentsStationary1974}.}
    \end{deflist}
\end{definition}

Note that the one-point conformal completion $\left(\tilde{S},\tilde{h}\right)$ in \autoref{def:asympflatGeroch} may not be unique. In \cite{gerochMultipoleMomentsII1970}, it is stated that $\left(\tilde{S},\tilde{h}\right)$ is unique up to conformal transformations $\varphi$ with $\varphi(i^0)=1$, but the proof is incorrect.\footnote{The topology defined in the proof is not a topology.} We provide a corrected proof of the following refined statement, as already stated in the introduction.

\uniqueness

In asymptotically flat spaces, it is not unreasonable to expect that they are diffeomorphic to $\R^3$ minus a ball on some region. In order to apply this theorem, one should think of $K$ as ``neighborhoods'' of points spoiling compactness of $\tilde{S}$. For example, if there are multiple asymptotically flat ends, we take one of them and the others are included in $K$. Also singularities (black holes) are included in $K$. The space might not be asymptotically flat there, but they still spoil compactness.

The proof comes down to proving: uniqueness up to homeomorphism, uniqueness up to diffeomorphism, uniqueness up to conformal transformation, and finally showing that the conformal factor at $i^0$ is fixed. The first step is based on the following lemma.

\begin{lemma}\label{lem:topman}
    Let $S$ be a topological $3$-manifold and let $K\subseteq S$ be a closed subset such that $S\setminus \Int{K}$ is not compact. Then there is at most one topology on $\tilde{S}=S\cup\{i^0\}$ for some point $i^0\notin S$ such that $\tilde{S}$ is a topological $3$-manifold, the inclusion $\iota\colon S\hookrightarrow \tilde{S}$ is an embedding and $\tilde{S}\setminus \Int{K}$ is compact.
\end{lemma}

\begin{proof}
    Suppose there is a topological manifold $\tilde{S}$ satisfying the conditions in the lemma. The idea of the proof is to view $\tilde{S}$ as the union of $S$ and $\tilde{S}\setminus K$, which both have a fixed topology. We proceed with the proof via four claims.

    \textbf{Claim 1.} \textit{A subset $V\subseteq \tilde{S}\setminus K$ is open if and only if either $V\subseteq S\setminus K$ is open or $i^0\in V$ and $\left(\tilde{S}\setminus \Int{K}\right)\setminus V$ is compact in $S\setminus \Int{K}$.}

    \textit{Proof of Claim 1.} First, we want to identify the open subsets of $\tilde{S}\setminus \Int{K}$ in the same way. Since $S\setminus \Int{K}$ is a locally compact Hausdorff space, it has a unique one-point compactification up to homeomorphism \cite[Theorem 29.1]{munkresTopology2014}. In particular, a subset $U\subseteq \tilde{S}\setminus \Int{K}$ is open if and only if either $U\subseteq S\setminus \Int{K}$ is open or $i^0\in U$ and $\left(\tilde{S}\setminus \Int{K}\right)\setminus U$ is compact in $S\setminus \Int{K}$.

    ``$\implies$": Let $V\subseteq \tilde{S}\setminus K$ be an open subset, then there exists an open subset $U\subseteq \tilde{S}\setminus \Int{K}$ such that $V=U\cap \left(\tilde{S}\setminus K\right)$. For $U$, there are two possibilities. Firstly, if $U\subseteq S\setminus \Int{K}$, then $V\subseteq S\setminus K$ is an open open subset. Secondly, suppose $i^0\in U$ and $\left(\tilde{S}\setminus \Int{K}\right)\setminus U$ is compact in $S\setminus \Int{K}$. Then we have $i^0\in V$. Since $\partial K\subseteq \tilde{S}\setminus \Int{K}$ is compact, and
    \[V=U\cap \left(\tilde{S}\setminus K\right) = U\setminus \partial K,\]
    it follows that
    \[\left(\tilde{S}\setminus \Int{K}\right)\setminus V=\left(\tilde{S}\setminus \Int{K}\right)\setminus \left(U\setminus \partial K\right) = \left(\left(\tilde{S}\setminus \Int{K}\right)\setminus U\right)\cup \partial K\]
    is compact.

    ``$\impliedby$": There are two cases to consider. For the first case, let $V\subseteq S\setminus K$ be an open subset, then $V\subseteq S\setminus \Int{K}$ is also open. But then $V\subseteq \tilde{S}\setminus \Int{K}$ is also open, from which we can conclude that $V\subseteq \tilde{S}\setminus K$ is an open subset. For the second case, let $V\subseteq \tilde{S}\setminus K$ be a subset containing $i^0$ and such that $\bigl(\tilde{S}\setminus \Int{K}\bigr)\setminus V$ is compact in $S\setminus \Int{K}$. Then we have that $V\subseteq \tilde{S}\setminus \Int{K}$ is open, which also gives that $V\subseteq \tilde{S}\setminus K$ is open.

    \textbf{Claim 2.} \textit{The family $\{S,\tilde{S}\setminus K\}$ of subsets of $\tilde{S}$ is an open cover of $\tilde{S}$.}

    \textit{Proof of Claim 2.} Since $K\subseteq S$, it is clear that $\tilde{S}=S\cup \bigl(\tilde{S}\setminus K\bigr)$. The singleton $\left\{i^0\right\}$ is closed in $\tilde{S}$ because $\tilde{S}$ is Hausdorff by assumption, so $S=\tilde{S}\setminus \left\{i^0\right\}$ is an open subset of $\tilde{S}$.

    We are left to show that $\tilde{S}\setminus K$ is an open subset of $\tilde{S}$. Let $\overline{K}$ denote the closure of $K$ in $\tilde{S}$, then we are done if $\overline{K}=K$. Since $S\setminus K$ is an open subset of $S$ and $S$ is open in $\tilde{S}$, the set $S\setminus K$ is also open in $\tilde{S}$. Therefore, $\overline{K}\subseteq K\cup \left\{i^0\right\}$. Let $U$ be a coordinate domain for $\tilde{S}$ centered at $i^0$. Since $\partial K$ is compact in $S$, it is also compact in $\tilde{S}$ and $U\setminus \partial K$ is open in $\tilde{S}$. Let $V$ be the connected component of $U\setminus \partial K$ containing $i^0$. Then $V$ is homeomorphic to an open, connected subset of $\R^3$. Moreover, $W=V\setminus \left\{i^0\right\}$ is also an open, connected subset of $\tilde{S}$. Hence, the set $W$ is open and connected in $S$, and does not intersect $\partial K$. Then $W\cap \Int{K}$ and $W\cap (S\setminus K)$ form a disjoint open cover of $W$, so by connectivity only one of them can be nonempty. Suppose $W\cap (S\setminus K) = \emptyset$, then $V\cap \bigl(\tilde{S}\setminus K\bigr) = \left\{i^0\right\}$. By construction of the subspace topology, $V\cap \bigl(\tilde{S}\setminus K\bigr)=\left\{i^0\right\}$ is open in $\tilde{S}\setminus K$, so $S\setminus \Int{K}$ is compact by Claim 1. But $S\setminus \Int{K}$ is homeomorphic $\R^3\setminus \mathbb{B}^3$, which is not compact, so we arrive at a contradiction. Therefore, we must have $W\cap (S\setminus K)\neq \emptyset$, implying that $W\cap \Int{K}=\emptyset$. Hence, $V$ is an open neighborhood of $i^0$ in $\tilde{S}$ that does not intersect $K$. We conlude that $i^0\notin \overline{K}$ and $\overline{K}=K$.

    \textbf{Claim 3.} \textit{Let $\mathcal{T}$ be the topology of $S$ and let $\mathcal{T}_{i^0}$ be the collection of open neighborhoods of $i^0$ in $\tilde{S}\setminus K$, then $\mathcal{T}\cup \mathcal{T}_{i^0}$ is a basis for a topology on $\tilde{S}$.}

    \textit{Proof of Claim 3.} We have $S\in \mathcal{T}$ and $\tilde{S}\setminus K\in \mathcal{T}_{i^0}$, and these open subsets of $\tilde{S}$ cover $\tilde{S}$ by Claim 2. Therefore, each point in $\tilde{S}$ is contained in an element of $\mathcal{T}\cup \mathcal{T}_{i^0}$. By definition of a basis for a topology on $\tilde{S}$, we are only left to show that for any $x\in U\cap V$ with $U,V\in \mathcal{T}\cup \mathcal{T}_{i^0}$, there exists a subset $W\in \mathcal{T}\cup\mathcal{T}_{i^0}$ such that $x\in W\subseteq U\cap V$ \cite[Section 2.13]{munkresTopology2014}. In particular, it suffices to show that $\mathcal{T}\cup \mathcal{T}_{i^0}$ is closed under taking intersections.

    There are a few cases to consider, depending on whether $U$ and $V$ belong to $\mathcal{T}$ or $\mathcal{T}_{i^0}$. If $U,V\in \mathcal{T}$, then $U\cap V\in \mathcal{T}$ because a topology is closed under taking intersections. If $U\in \mathcal{T}$ and $V\in \mathcal{T}_{i^0}$, we have $U\cap V = U\cap \left(V\setminus \left\{i^0\right\}\right)$. Since $\left\{i^0\right\}$ is closed in $\tilde{S}\setminus K$, the set $V\setminus \left\{i^0\right\}$ must be open in $\tilde{S}\setminus K$, but then $V\setminus \left\{i^0\right\}$ is open in $S\setminus K$ by Claim 1, so it is open in $S$. Therefore, $U\cap V= U\cap \left(V\setminus \left\{i^0\right\}\right)\in \mathcal{T}$. Finally, if $U,V\in \mathcal{T}_{i^0}$, we have $i^0\in U\cap V$ and
    \[\bigl(\tilde{S}\setminus \Int{K}\bigr)\setminus (U\cap V) = \left(\bigl(\tilde{S}\setminus \Int{K}\bigr)\setminus U\right) \cup \left(\bigl(\tilde{S}\setminus \Int{K}\bigr)\setminus V\right),\]
    which is compact because it is a union of two compact sets. Hence, $U\cap V\in \mathcal{T}_{i^0}$ by Claim 1. We conclude that $\mathcal{T}\cup \mathcal{T}_{i^0}$ is closed under all possible intersections and it is a basis for a topology on $\tilde{S}$.

    \textbf{Claim 4.} \textit{The topology of $\tilde{S}$ is the topology generated by $\mathcal{T}\cup \mathcal{T}_{i^0}$.}

    \textit{Proof of Claim 4.} Let $\tilde{\mathcal{T}}$ be the topology of $\tilde{S}$. The collection $\mathcal{T}\cup \mathcal{T}_{i^0}$ of subsets of $\tilde{S}$ consists of subsets that are either open in $S$ or in $\tilde{S}\setminus K$. Since $S$ and $\tilde{S}\setminus K$ are open in $\tilde{S}$ by Claim 2, these subsets must also be open in $\tilde{S}$. Hence, $\mathcal{T}\cup \mathcal{T}_{i^0}\subseteq \tilde{\mathcal{T}}$, from which we conclude that the topology generated by $\mathcal{T}\cup \mathcal{T}_{i^0}$ must be contained in $\tilde{\mathcal{T}}$.

    Conversely, let $U\in \tilde{\mathcal{T}}$. If $i^0\notin U$, then we have $U=U\cap S\in \mathcal{T}$. If $i^0\in U$, then $U\cap S \in \mathcal{T}$ and $U\cap \bigl(\tilde{S}\setminus K\bigr)\in \mathcal{T}_{i^0}$ because $S, \tilde{S}\setminus K \subseteq \tilde{S}$ are open. But we also have
    \[U=(U\cap S)\cup \left(U\cap \bigl(\tilde{S}\setminus K\bigr)\right),\]
    so $U$ is contained in the topology generated by $\mathcal{T}\cup \mathcal{T}_{i^0}$. Hence, $\tilde{\mathcal{T}}$ equals the topology generated by $\mathcal{T}\cup \mathcal{T}_{i^0}$.

    To conclude the proof of the lemma, Claim 1 fixes $\mathcal{T}_{i^0}$ as the collection of subsets $V\subseteq \tilde{S}\setminus K$ such that $i^0\in V$ and $\bigl(\tilde{S}\setminus \Int{K}\bigr)\setminus V$ is compact in $S\setminus \Int{K}$. Since the topology on $S$ and the subset $K$ are given, it fixes both $\mathcal{T}$ and $\mathcal{T}_{i^0}$. The topology on $\tilde{S}$ is fixed by Claim 4.
\end{proof}

\begin{proof}[Proof of \autoref{thm:uniqueness}]
    \autoref{lem:topman} tells us that if $(S,h)$ is asymptotically flat, and $\left(\tilde{S},\tilde{h}\right)$ and $\Omega$ are as in \autoref{def:asympflatGeroch}, then there is a unique topology on $\tilde{S}$ such that $\tilde{S}$ is a topological $3$-manifold, $\iota\colon S\hookrightarrow \tilde{S}$ is an embedding, and $\tilde{S}\setminus \Int{K}$ is compact. Moise's theorem \cite{moiseAffineStructures3Manifolds1952} tells us that every topological $3$-manifold admits, up to diffeomorphism, a unique smooth structure. Asymptotic flatness implies that any two smooth structures on $\tilde{S}$ must agree on $S\subseteq \tilde{S}$, so that there is a diffeomorphism preserving $i^0$. Hence, $\tilde{S}$ is unique up to diffeomorphisms preserving $i^0$. We are left to show uniqueness up to conformal transformations and to fix the conformal factor at $i^0$.

    \textbf{Uniqueness of $\bigl(\tilde{S}, \tilde{h}\bigr)$ up to conformal transformations.} Assume that we have two metrics $\tilde{h}_1$ and $\tilde{h}_2$ on $\tilde{S}$ with conformal factors $\Omega_1$ and $\Omega_2$, respectively, satisfying the conditions (i) and (ii) of \autoref{def:asympflatGeroch}. On $S$, we have $h=\iota^*\left(\Omega_1^{-2}\tilde{h}_1\right) = \iota^*\left(\Omega_2^{-2}\tilde{h}_2\right)$. The functions $\Omega_1$ and $\Omega_2$ are smooth and nonvanishing on $S$, so $\alpha= \Omega_2/\Omega_1$ is a well-defined, smooth function on $S$. Moreover, $\tilde{h}_2= \left(\Omega_2/\Omega_1\right)^2\tilde{h}_1 = \alpha^2\tilde{h}_1$ on $S$. It remains to extend this property to $\tilde{S}=S\cup\left\{i^0\right\}$. Let $\left(E_1, E_2, E_3\right)$ be an orthonormal frame on an open neighborhood $U$ of $i^0$ with respect to $\tilde{h}_1$. Then we have
    \[\alpha^2 = \alpha^2\tilde{h}_1\left(E_1, E_1\right) = \tilde{h}_2\left(E_1,E_1\right),\]
    on $U\setminus \left\{i^0\right\}$. The right-hand side is a smooth, (strictly) positive function on $U$, so $\alpha^2$ also extends smoothly to $i^0$ with a positive value. Therefore, $\alpha$ also extends to a smooth, nonvanishing function on $\tilde{S}$. By continuity, we must have $\tilde{h}_2 = \alpha^2 \tilde{h}_1$ on all of $\tilde{S}$, establishing uniqueness up to conformal transformations.

    \textbf{Uniqueness of the conformal factor at $i^0$.} Let us compare the two metrics and conformal factors in light of condition (iii) of \autoref{def:asympflatGeroch}. By the previous paragraph, we have $\Omega_2 = \alpha \Omega_1$ for some nonvanishing, smooth function $\alpha$ on $\tilde{S}$. Let $\tilde{D}_i$ denote the Levi-Civita connection with respect to $\tilde{h}_i$, for $i=1,2$. Then the relation for the Levi-Civita connection between conformal metrics \cite[Proposition 7.29]{leeIntroductionRiemannianManifolds2018} gives
    \[\tilde{D}_2\left(d \Omega_2\right) = \tilde{D}_1\left(d \Omega_2\right) - \alpha^{-1} \left(d \Omega_2 \otimes d\alpha  +d\alpha\otimes d\Omega_2\right) + \alpha^{-1} d\Omega_2\left(\grad_{\tilde{h}_1}{\alpha}\right) \tilde{h}_1.\]
    When evaluating at $i^0$, the last three terms vanish because $\left.d\Omega_2\right\rvert_{i^0}=0$ by condition (iii) in \autoref{def:asympflatGeroch}. The first term is
    \[\tilde{D}_1\left(d\Omega_2\right) = \tilde{D}_1\left(d\left(\alpha\Omega_1\right)\right) = \tilde{D}_1\left(\alpha d\Omega_1 + \Omega_1 d\alpha\right) = \alpha\tilde{D}_1\left(d\Omega_1\right) + d\alpha \otimes d\Omega_1 + d\Omega_1 \otimes d\alpha + \Omega_1\tilde{D}_1\left(d\alpha\right),\]
    of which the last three terms also vanish at $i^0$ because of condition (iii) in \autoref{def:asympflatGeroch}. So,
    \[\left.\tilde{D}_2\left(d \Omega_2\right)\right\rvert_{i^0} = \alpha\left(i^0\right)\left.\tilde{D}_1\left(\alpha d\Omega_1\right)\right\rvert_{i^0},\]
    and applying condition (iii) in \autoref{def:asympflatGeroch} once more yields
    \[2 \left(\alpha\left(i^0\right)\right)^2 \left.\tilde{h}_1\right\rvert_{i^0}=2\left.\tilde{h}_2\right\rvert_{i^0} = \left.\tilde{D}_2\left(d \Omega_2\right)\right\rvert_{i^0} = \alpha\left(i^0\right)\left.\tilde{D}_1\left(d\Omega_1\right)\right\rvert_{i^0} = 2\alpha\left(i^0\right)\left.\tilde{h}_1\right\rvert_{i^0}.\]
    Since $\alpha$ is nonvanishing on $\tilde{S}$, this is only possible if $\alpha\left(i^0\right)=1$.
\end{proof}

Note that Geroch's original approach \cite{gerochMultipoleMomentsII1970} in proving uniqueness of the conformal completion does not involve Moise's theorem. He beautifully fixes the smooth structure by determining the smooth functions using conformal Laplacians. This approach depends less heavily on the dimension and is more constructive.

\subsection{Einstein equations}\label{sec:EEs}

Multipole moments are constructed from specific potentials, which themselves are derived from the Einstein equations. We discuss how to find these potentials on $S$, which are then extended to $\tilde{S}$ (from \autoref{def:asympflatGeroch}).

Recall that, in vacuum, the twist one-form \eqref{eq:twistoneform} is closed. Let us now restrict $S$ (and $M$ accordingly) such that it is diffeomorphic to $\mathbb{B}^3\setminus \{0\}$. If $S$ is asymptotically flat in the sense of \autoref{def:asympflatGeroch}, we can take a coordinate ball $B$ for $\tilde{S}$ centered at $i^0$, and we restrict $\tilde{S}$ to $B$ and $S$ to $B\setminus \{i^0\}$. Since we are interested in the local behaviour around $i^0$, we do not lose any information by doing so. In particular, the first de Rham cohomology of $S$ now vanishes because $\mathbb{B}^3\setminus \{0\}$ is homotopy equivalent to $\Sph^2$. Therefore, the twist one-form is not only closed, but also exact, i.e.,
\begin{equation}
    \omega=df,
\end{equation}
for some $f\in C^\infty(S)$, called the \textit{twist potential}. In vacuum, Hansen \cite{hansenMultipoleMomentsStationary1974} introduced the potentials
\begin{equation}\label{eq:masspotential}
    \phi_M =\frac{1-\lambda^2 -f^2}{4\lambda},
\end{equation}
and
\begin{equation}\label{eq:angmompotential}
    \phi_J = \frac{-f}{2\lambda},
\end{equation}
called the \textit{mass potential} and the \textit{angular momentum} (or \textit{spin}, or \textit{current}) \textit{potential}, respectively.

\begin{definition}\label{def:asympflatvac}
    A stationary spacetime $(M,g,\xi)$ is called an \textit{asymptotically flat vacuum solution} if the observer space $(S,h)$ is asymptotically flat, $(M,g)$ is a solution of the Einstein equations in vacuum, and the functions
    \begin{equation}\label{eq:phitilde}
        \tilde{\phi}_M=\Omega^{-\frac12}\phi_M \qquad \text{and} \qquad \tilde{\phi}_J=\Omega^{-\frac12}\phi_J
    \end{equation}
    extend to smooth functions on $\tilde{S}$.
\end{definition}

It is possible to relax the smoothness condition in the definition above to, for example, $C^2$. As discussed by Hansen \cite{hansenMultipoleMomentsStationary1974}, the regularity of the potentials can be improved using elliptic partial differential equations.

\section{Multipole moments}\label{sec:mms}

The coefficients of a Taylor expansion of an analytic field can be found by evaluating consecutive derivatives. Multipole moments are constructed in a similar way. We impose the three assumptions from \autoref{sec:assumptions}, culminating in a stationary asymptotically flat vacuum solution as in \autoref{def:asympflatvac}. This yields the transformed gravitational potentials $\tilde{\phi}_i$ on $\tilde{S}$ which must be smooth at $i^0$, and we can compute an asymptotic expansion of each $\tilde{\phi}_i$ at $i^0$.

\subsection{Definition}

Like in Newtonian gravity, multipole moments are symmetric trace-free tensors. Such tensors provide an alternative description for spherical harmonics \cite{poissonGravityNewtonianPostNewtonian2014}. A short review of such tensors is contained in \autoref{app:STF}.

\begin{definition}\label{def:seqtensors}
    Let $\left(S,h\right)$ be an asymptotically flat Riemannian manifold with $\left(\tilde{S},\tilde{h}\right)$ and $\Omega$ as in \autoref{def:asympflatGeroch}. Let $\phi$ be a smooth function on $S$ such that $\tilde{\phi}=\Omega^{-\frac12}\phi$ extends smoothly to $\tilde{S}=S\cup \{i^0\}$. The \textit{sequence $\left(P^k\right)_{k\in \N_0}$ of symmetric trace-free covariant $k$-tensor fields induced by $\phi$} on $\tilde{S}$ is inductively defined by $P^0=\tilde{\phi}$ and
    \begin{equation}\label{eq:GHmoments_step}
        P^{k+1} = \left(\tilde{D} P^k - \frac12 k (2k-1) P^{k-1} \otimes \tilde{Rc} \right)^{STF},
    \end{equation}
    for $k\in \N_0:=\N\cup \{0\}$, where $T^{STF}$ denotes taking the totally symmetric and trace-free part of $T$, and $\tilde{Rc}$ is the Ricci tensor on $\left(\tilde{S},\tilde{h}\right)$. The \textit{$2^k$-pole moment of $\phi$} is $\left.P^{k}\right\rvert_{i^0}$.
\end{definition}

In \eqref{eq:GHmoments_step}, we recognise taking derivatives of the transformed potentials $\tilde{\phi}$ using the Levi-Civita connection, but the correction term with the Ricci tensor may come as a surprise. The origin of this term comes from the transformation law under the residual conformal transformations and is discussed in \autoref{sec:transformationlaw}.

\autoref{def:seqtensors} does not depend on the Einstein equations. However, we use the Einstein equations to define the potentials $\phi$ to which we apply the construction. In vacuum, the potentials are the mass and angular momentum potentials $\phi_M$ and $\phi_J$ defined in \eqref{eq:masspotential} and \eqref{eq:angmompotential}.

\begin{definition}\label{def:GHmoments}
    Let $(S, h)$ be an asymptotically flat Riemannian manifold whose one-point conformal completion is $\bigl(\tilde{S}, \tilde{h}\bigr)$. Let $\phi_M$ and $\phi_J$ be the mass and angular momentum potential, respectively, and suppose $\tilde{\phi}_A=\Omega^{-\frac12}\phi_A$ extends to a smooth function on $\tilde{S}$ for $A=M,J$. The \textit{mass $2^k$-pole moment} is the $2^k$-pole moment of $\phi_M$ and is denoted by $M^k$, and the \textit{angular momentum $2^k$-pole moment} is the $2^k$-pole moment of $\phi_J$ and is denoted by $J^k$.
\end{definition}

\begin{example}
    The most important nontrivial example for which we can calculate multipole moments is the Kerr spacetime. However, even for the Kerr spacetime it is difficult to perform the calculation to arbitrary order. Therefore, we only state the result. Since the Kerr spacetime is axisymmetric, the calculations can be simplified greatly using algorithms by Fodor, Hoenselaers and Perjés \cite{fodorMultipoleMomentsAxisymmetric1989} and Bäckdahl and Herberthson \cite{backdahlExplicitMultipoleMoments2005}. We refer to the latter for a precise calculation of the multipole moments for the Kerr spacetime up to arbitrary order. As already noted by Hansen \cite{hansenMultipoleMomentsStationary1974}, there is an axis vector field $\tilde{Z}$ on $\tilde{S}$ such that the multipole moments $\left.P^k\right\rvert_{i^0}$ can be reconstructed from the constants
    \begin{equation}\label{eq:constantsmmsaxisym}
        p^k = \frac{1}{k!} P^k\left(\tilde{Z}, \dots, \tilde{Z}\right)
    \end{equation}
    via
    \begin{equation}
        P^k = (2k-1)!! p^k \left(\tilde{Z}^\flat\otimes \cdots \tilde{Z}^\flat\right)^{STF},
    \end{equation}
    where $\tilde{Z}^\flat = \tilde{h}\left(\tilde{Z},\cdot\right)$. In vacuum, we denote these constants $p^k$ by $m^k$ and $j^k$ for the mass and angular momentum multipole moments, respectively. The multipole moments of the Kerr spacetime turn out to be
    \begin{equation}\label{eq:mmsKerr}
        m^{2k} = (-1)^k m a^{2k}, \qquad m^{2k+1}=0, \qquad j^{2k}=0, \qquad j^{2k+1}=(-1)^k m a^{2k+1},
    \end{equation}
    for $k\in \N_0$. In particular, we have $m^0=m$,  $j^0=0$, $m^1=0$ and $j^1=ma$, precisely as one would expect for the mass and angular momentum.
\end{example}

\subsection{Transformation law}\label{sec:transformationlaw}

The recursively defined covariant tensor fields $P^k$ in \autoref{def:seqtensors} are defined on $\left(\tilde{S},\tilde{h}\right)$. According to \autoref{thm:uniqueness}, this space is unique up to certain conformal transformations of $\left(\tilde{S},\tilde{h}\right)$. Beig \cite{beigMultipoleExpansionGeneral1981} sketched for the first time how the multipole moments transform under conformal changes of $\tilde{h}$, in the setting of static vacuum spacetimes. Here, we present a more general result for stationary spacetimes which is based on Beig's original approach. In the proof, we utilise some identities for symmetric trace-free tensors. We refer to \autoref{app:STF} for these results.

\begin{theorem}\label{prop:mms_conftransglob}
    Let $\left(S,h\right)$ be an asymptotically flat Riemannian manifold with one-point extension $\left(\tilde{S},\tilde{h}_1\right)$ and conformal factor $\Omega_1$. Let $\alpha$ be a smooth positive function on $\tilde{S}$ with $\alpha\left(i^0\right)=1$ and let $\phi$ be a smooth function on $S$ such that $\tilde{\phi}_1=\Omega_1^{-\frac12} \phi$ extends to a smooth function on $\tilde{S}$. Let $\tilde{h}_2 = \alpha^2\tilde{h}_1$, and let $\left(P_1^k\right)$ and $\left(P_2^k\right)$ be the sequence of symmetric trace-free covariant $k$-tensor fields of $\phi$ of \autoref{def:seqtensors} with respect to $\tilde{h}_1$ and $\tilde{h}_2$, respectively. Then
    \begin{equation}\label{eq:mms_conftransglob}
        P_2^k = \sum_{m=0}^{k} \binom{k}{m} \frac{(2k-1)!!}{(2m-1)!!}(-2)^{-(k-m)} \alpha^{-\frac12 - (k-m)} \left(P_1^m \otimes d\alpha^{\otimes (k-m)}\right)^{STF},
    \end{equation}
    where $d\alpha^{\otimes n}=d\alpha\otimes \cdots \otimes d\alpha$, the tensor product of $n$ $d\alpha$'s and the double factorial is defined by $(-1)!!=1$ and $(2n-1)!!=(2n-1)(2n-3)\cdots 1$ for $n\in \N$.
\end{theorem}

\begin{remark}
    Note that it does not matter whether we take the (symmetric) trace-free part $\left(\cdot\right)^{STF}$ with respect to $\tilde{h}_1$ or $\tilde{h}_2$. A tensor is trace-free with respect to one of them if and only if it is with respect to the other. Alternatively, we can see this because replacing $h$ in \eqref{eq:STFpart} by $\tilde{h}_1$ and $\tilde{h}_2$ yield the same result as the factors $\alpha$ cancel each other.
\end{remark}

\begin{proof}
    We prove the result by induction. Let $\tilde{D}_1$ and $\tilde{D}_2$ denote the Levi-Civita connections and let $\tilde{Rc}_1$ and $\tilde{Rc}_2$ denote the Ricci tensors with respect to $\tilde{h}_1$ and $\tilde{h}_2$, respectively. Since $\Omega_2= \alpha\Omega_1$, we have $\tilde{\phi}_2= \Omega_2^{-\frac12} \phi = \tilde{\phi}_2 = \alpha^{-\frac12}\tilde{\phi}_1$, and this also extends to a smooth function on $\tilde{S}$ because $\alpha(i^0)=1$. Following \autoref{def:seqtensors}, we have $P_1^0 = \tilde{\phi}_1$ and $P_2^0=\tilde{\phi}_2$, so
    \[P_2^0 =\alpha^{-\frac12} P_1^0.\]
    Moreover,
    \[P_2^1= \tilde{D}_2 P_2^0 = d P_2^0 = \alpha^{-\frac12} d P_1^0 -\frac12\alpha^{-\frac32}P_1^0 d\alpha = \alpha^{-\frac12} P_1^1 - \frac12 \alpha^{-\frac32} P_1^0 d\alpha, \]
    proving \eqref{eq:mms_conftransglob} for $k=0,1$.

    Assume \eqref{eq:mms_conftransglob} is satisfied for $k-1$ and $k$ for some $k\in \N$. We want to calculate $P_2^{k+1}$ using \eqref{eq:GHmoments_step}, so we need $\tilde{D}_2 P^k_2$ and $P^{k-1}_2 \otimes \tilde{Rc}_2$. Under conformal transformations, the Levi-Civita connection on covariant $k$-tensor fields transforms as \cite[Proposition 7.29]{leeIntroductionRiemannianManifolds2018}
    \[\begin{split} \tilde{D}_2 P_2^k(X_1,\dots, X_{k+1}) &= \tilde{D}_1 P_2^k(X_1,\dots, X_{k+1}) - k \alpha^{-1}X_{k+1}(\alpha) P_2^k(X_1,\dots, X_{k}) \\&\qquad- \sum_{i=1}^k \alpha^{-1}X_i(\alpha) P_2^k(X_1,\dots, X_{i-1},X_{k+1},X_{i+1},\dots, X_k) \\&\qquad+ \sum_{i=1}^k \alpha^{-1} h(X_{k+1},X_i) P_2^k(X_1,\dots, X_{i-1},\grad_h{\alpha},X_{i+1},\dots, X_k).\end{split}\]
    When we take the symmetric trace-free part of $\tilde{D}_2 P_2^k$, the last summation vanishes, so
    \[\left(\tilde{D}_2 P_2^k\right)^{STF} = \left(\tilde{D}_1 P_2^k\right)^{STF} - 2k \alpha^{-1}\left(P_2^k\otimes d\alpha\right)^{STF}.\]
    By the induction hypothesis, this gives
    \begin{equation}\label{eq:D2P2k}
        \begin{split} \left(\tilde{D}_2 P_2^k\right)^{STF} &= \sum_{m=0}^{k} \binom{k}{m} \frac{(2k-1)!!}{(2m-1)!!}(-2)^{-(k-m)} \alpha^{-\frac12 - (k-m)}  \\&\qquad \cdot\left(\left(\tilde{D}_1 P_1^m \otimes d\alpha^{\otimes (k-m)}\right)^{STF}\right. \\&\qquad \qquad \left. + (k-m) \left(P_1^m \otimes \tilde{D}_1(d\alpha)\otimes d\alpha^{\otimes (k-1-m)}\right)^{STF} \right. \\&\qquad \qquad \left.-\frac12 \left(6k-2m+1\right)  \alpha^{-1} \left(P_1^m \otimes d\alpha^{\otimes (k+1-m)}\right)^{STF} \right),\end{split}
    \end{equation}
    where we utilised \eqref{eq:tensorprodSTF} and \eqref{eq:covdevSTF} to simplify the symmetric trace-free parts. The Ricci tensor transforms as \cite[Theorem 7.30]{leeIntroductionRiemannianManifolds2018}
    \[\tilde{Rc}_2= \tilde{Rc}_1  - \alpha^{-1}\tilde{D}_1\left(d\alpha\right) - \alpha^{-1} \left(\tilde{\Delta}_{\tilde{h}_1}\alpha\right) \tilde{h}_1 + 2 \alpha^{-2} d\alpha \otimes d\alpha,\]
    and taking the symmetric trace-free part gives
    \[\left(\tilde{Rc}_2\right)^{STF} = \left(\tilde{Rc}_1\right)^{STF} - \alpha^{-1} \left(\tilde{D}_1\left(d\alpha\right)\right)^{STF} + 2\alpha^{-2} \left(d\alpha\otimes d\alpha\right)^{STF},\]
    using \eqref{eq:tensorprodmetricSTF}. By the induction hypothesis for $k-1$,
    \begin{equation}\label{eq:P2k-1Rc}
        \begin{split}\left(P_2^{k-1}\otimes \tilde{Rc}_2\right)^{STF}&= \sum_{m=0}^{k-1} \binom{k-1}{m} \frac{(2k-3)!!}{(2m-1)!!}(-2)^{-(k-1-m)} \alpha^{-\frac12 - (k-1-m)}  \\&\qquad \cdot\left(\left(P_1^{m}\otimes \tilde{Rc}_1\otimes d\alpha^{\otimes(k-1-m)}\right)^{STF} \right.\\&\qquad\qquad\left.- \alpha^{-1}\left(P_1^m\otimes \tilde{D}_1(d\alpha)\otimes d\alpha^{\otimes(k-1-m)}\right)^{STF} \right.\\&\qquad\qquad\left.+2\alpha^{-2}\left(P_1^m\otimes d\alpha^{\otimes(k+1-m)}\right)^{STF} \right),\end{split}
    \end{equation}
    where we used \eqref{eq:tensorprodSTF} to simplify the expression. For $P_2^{k+1}$, following \eqref{eq:GHmoments_step}, we have
    \[ P_2^{k+1} = \left(\tilde{D}_2 P_2^k - \frac12 k(2k-1) P_2^{k-1}\otimes \tilde{Rc}_2\right)^{STF} =A+B+C,\]
    where, from \eqref{eq:D2P2k} and \eqref{eq:P2k-1Rc}, $A$ contains the terms with $\tilde{D}_1 P_1^m$ and $P_1^m\otimes \tilde{Rc}_1$, $B$ contains the terms with $\tilde{D}_1(d\alpha)$, and $C$ contains the other terms which are of the form $P_1^m\otimes d\alpha^{\otimes (k+1-m)}$. Rewriting a little bit easily shows that $B=0$ because for each $m=0,\dots, k-1$ the coefficients cancel:
    \[\binom{k}{m} \frac{(2k-1)!!}{(2m-1)!!} (-2)^{-(k-m)}(k-m) - \frac12 k(2k-1) \binom{k-1}{m} \frac{(2k-3)!!}{(2m-1)!!} (-2)^{-(k-m-1)}(-1)=0.\]
    For $A$, shifting the summation for the $\tilde{Rc}_1$-terms gives
    \[\begin{split} A&= \sum_{m=0}^k \binom{k}{m} \frac{(2k-1)!!}{(2m-1)!!} (-2)^{-(k-m)}\alpha^{-\frac12-(k-m)}\left(\tilde{D}_1 P_1^m \otimes d\alpha^{\otimes (k-m)}\right)^{STF} \\&\qquad-\frac12 k \sum_{m=0}^{k-1} \binom{k-1}{m} \frac{(2k-1)!!}{(2m-1)!!} (-2)^{-(k-1-m)} \alpha^{-\frac12 -(k-1-m)}\left(P_1^m \otimes \tilde{Rc}_1\otimes d\alpha^{\otimes (k-1-m)}\right)^{STF} \\&= \sum_{m=0}^k \binom{k}{m} \frac{(2k-1)!!}{(2m-1)!!} (-2)^{-(k-m)}\alpha^{-\frac12-(k-m)}\left(\tilde{D}_1 P_1^m \otimes d\alpha^{\otimes (k-m)}\right)^{STF} \\&\qquad-\frac12 k \sum_{m=1}^{k} \binom{k-1}{m-1} \frac{(2k-1)!!}{(2m-3)!!} (-2)^{-(k-m)} \alpha^{-\frac12 -(k-m)}\left(P_1^{m-1} \otimes \tilde{Rc}_1\otimes d\alpha^{\otimes (k-m)}\right)^{STF}.\end{split}\]
    Taking the summations together and exploiting \eqref{eq:tensorprodSTF} gives
    \[\begin{split} A &= \sum_{m=0}^l \binom{k}{m} \frac{(2k-1)!!}{(2m-1)!!} (-2)^{-(k-m)}\alpha^{-\frac12-(k-m)} \\&\qquad \cdot\left(\left(\tilde{D}_1 P_1^m - \frac12 m(2m-1) P_1^{m-1}\otimes\tilde{Rc}_1\right)^{STF}\otimes d\alpha^{\otimes (k-m)}\right)^{STF}\end{split}\]
    Using \eqref{eq:GHmoments_step} for $m$ and shifting the summation again, we find
    \[\begin{split} A&= \sum_{m=0}^k \binom{k}{m} \frac{(2k-1)!!}{(2m-1)!!} (-2)^{-(k-m)}\alpha^{-\frac12-(k-m)} \left(P_1^{m+1}\otimes d\alpha^{\otimes (k-m)}\right)^{STF} \\&= \sum_{m=1}^{k+1} \binom{k}{m-1} \frac{(2k-1)!!}{(2m-3)!!} (-2)^{-(k+1-m)}\alpha^{-\frac12-(k+1-m)} \left(P_1^{m}\otimes d\alpha^{\otimes (k+1-m)}\right)^{STF}. \end{split}\]
    Finally, for $C$, we have
    \[\begin{split} C&= \sum_{m=0}^k \binom{k}{m} \frac{(2k-1)!!}{(2m-1)!!} (-2)^{-(k+1-m)} \left(2(k+m)+1\right) \alpha^{-\frac12 - (k+1-m)}\left(P_1^m\otimes d\alpha^{\otimes(k+1-m)}\right)^{STF}. \end{split}\]
    We want to add $A$ and $C$, giving three type of terms: the $m=0$ term in $C$, the $m=k+1$ term in $A$ and the terms for $m=1,\dots, k$. For the latter, an easy calculation yields
    \[\binom{k}{m-1} (2m-1) + \binom{k}{m} (2(k+m)+1) = \binom{k+1}{m}(2k+1).\]
    Therefore,
    \[\begin{split} P_2^{l+1} &= A+C
    \\&= (2k+1)!! (-2)^{-(k+1)}\alpha^{-\frac12 - (k+1)}\left(P_1^0\otimes d\alpha^{\otimes(k+1)}\right)^{STF} + \alpha^{-\frac12} \left(P_1^{k+1}\right)^{STF} \\&\qquad +\sum_{m=1}^{k} \binom{k+1}{m} \frac{(2k+1)!!}{(2m-1)!!} (-2)^{-(k+1-m)} \alpha^{-\frac12 - (k+1-m)} \left(P_1^m\otimes d\alpha^{\otimes(k+1-m)}\right)^{STF} \\&=\sum_{m=0}^{k+1} \binom{k+1}{m} \frac{(2k+1)!!}{(2m-1)!!} (-2)^{-(k+1-m)} \alpha^{-\frac12 - (k+1-m)} \left(P_1^m\otimes d\alpha^{\otimes(k+1-m)}\right)^{STF},\end{split}\]
    proving \eqref{eq:mms_conftransglob} by induction.
\end{proof}

\begin{corollary}\label{cor:mms_conftransi0}
    In the setting of \autoref{prop:mms_conftransglob}, the multipole moments transform as
    \begin{equation}\label{eq:mms_conftransi0}
        \left.P_2^k\right\rvert_{i^0} = \sum_{m=0}^{k} \binom{k}{m} \frac{(2k-1)!!}{(2m-1)!!}(-2)^{-(k-m)}  \left(\left.P_1^m\right\rvert_{i^0} \otimes \left.d\alpha\right\rvert_{i^0}^{\otimes (k-m)}\right)^{STF}.
    \end{equation}
\end{corollary}

\begin{proof}
    The result follows immediately from evaluating equation \eqref{eq:mms_conftransglob} at $i^0$ as $\alpha(i^0)=1$.
\end{proof}

Equation \eqref{eq:mms_conftransi0} determines how the multipole moments $\left.P^k\right\rvert_{i^0}$ behave under the residual conformal transformations. Viewing \eqref{eq:mms_conftransi0} as a perturbative expansion in $\left.d\alpha\right\rvert_{i^0}$, we see that the first-order correction of the $2^k$-multipole moment is proportional to $\mathcal{C}\left(\left.P_1^{k-1}\right\rvert_{i^0} \otimes \left.d\alpha\right\rvert_{i^0}\right)$, i.e., only depends on the $2^{k-1}$-multipole moment. Newtonian multipole moments behave in the same way when displacing the origin \cite{gerochMultipoleMomentsFlat1970,gerochMultipoleMomentsII1970}. At first, the $\tilde{Rc}$-term in the definition of the covariant tensors $P^k$ in \eqref{eq:GHmoments_step} might seem surprising as the coefficients in an expansion are usually obtained by differentiation. However, this term precisely cancels the correction term that is proportional to $\mathcal{C}\left(\left.P_1^{k-2}\right\vert_{i^0} \otimes \left.\tilde{D}_1(d\alpha)\right\rvert_{i^0}\right)$.

In Newtonian gravity, we often pick the origin such that it lies at the center of mass. That means, if the mass $2^k$-pole moments is the first one that is nonvanishing, then the mass $2^{k+1}$-pole moment does vanish. Assume that the mass of the system is nonvanishing, then the mass monopole moment $M^0$ is nonzero. We can apply a suitable conformal transformation such that the mass dipole moment vanishes. \autoref{cor:mms_conftransi0} yields
\[M^1_2 =M^1_1 -\frac12   M_0^1 \left.d\alpha\right\rvert_{i^0},\]
so we want to take $\alpha$ such that
\[\left.d\alpha\right\rvert_{i^0} = \frac{2}{M_0^1} M^1_1.\]
This is always possible and there is still some freedom left in $\alpha$. However, there is no freedom left in the multipole moments anymore. The multipole moments only change by $\left.d\alpha\right\rvert_{i^0}$, so choosing a conformal factor in such a way fixes the multipole moments.

\section{Discussion}\label{sec:discussion}

In this paper, we reviewed the construction of the Geroch--Hansen multipole moments and the assumptions on the spacetime for ensuring a well-defined concept. There are some areas where potential generalisations and extensions could be explored.

First, the Geroch--Hansen formalism is limited to stationary spacetimes, where we required the associated stationary vector field to be complete (\autoref{def:stationary-complete}). One might wonder what happens if we drop the completeness assumption. As discussed above \autoref{def:obsspace}, the observer space might fail to be Hausdorff. It would be interesting to check whether the completeness assumption is necessary or if it can be dropped.

Second, we defined asymptotic flatness for three-dimensional Riemannian manifolds in \autoref{def:asympflatGeroch}, whether or not they as observer spaces of stationary spacetimes. In principle, we could drop the dimensionality condition. For instance, the $n$-sphere serves as a one-point compactification of $\R^n$, suggesting parallels in higher-dimensional settings. While our proof of the uniqueness result in \autoref{thm:uniqueness} critically relies on three-dimensionality, as remarked below the proof of the theorem, Geroch's technique using the conformal Laplacian \cite{gerochMultipoleMomentsII1970} offers a way to generalise the result to Riemannian manifolds of dimensions $n>2$. This extends Geroch's asymptotic flatness condition to higher dimensions and one can broaden the applicability of the Geroch--Hansen formalism to higher dimensions.

Additionaly, the vacuum assumption could be relaxed. The Geroch--Hansen formalism can still be applied spacetimes with matter, but we have to apply \autoref{def:GHmoments} to different potentials. Simon \cite{simonMultipoleExpansionStationary1984} already constructed suitable potentials in electrovacuum. In that case, we do not only have mass and angular momentum potentials, but also electric and magnetic potentials, and the formalism should also be applied to the latter. These additional potentials yield so-called matter multipole moments. The mass and angular momentum potentials also had to be adapted because they are based on closedness of the twist one-form as discussed in \autoref{sec:EEs}, which is a result of vacuum assumption. Mayerson \cite{mayersonGravitationalMultipolesGeneral2023} recently constructed potentials for rather general solutions of the Einstein equations. However, his approach does not currently extend to matter potentials, leaving a gap in the formulation of corresponding multipole moments.

The significance of the Geroch--Hansen multipole moments in general relativity lies in their ability to characterise a spacetime. Beig and Simon \cite{beigMultipoleExpansionStationary1981}, as well as Kundu \cite{kunduAnalyticityStationaryGravitational1981}, established that the multipole moments uniquely determine the local structure near spatial infinity ($i^0$) in vacuum spacetimes. Simon \cite{simonMultipoleExpansionStationary1984} extended this result to the electrovacuum case, where the matter multipole moments play a critical role. For example, the Kerr and Kerr--Newman spacetimes share identical mass and angular momentum multipole moments but differ in their electric and magnetic multipole moments. This limitation is evident in Mayerson's construction \cite{mayersonGravitationalMultipolesGeneral2023}. Constructing matter potentials for more general spacetimes and evaluating whether the resulting multipole moments characterise the local structure near $i^0$ would be a critical step forward.

\begin{appendices}
\section{Symmetric trace-free tensors}\label{app:STF}

In this appendix, we briefly discuss how to construct the symmetric trace-free part of a tensor. Let $h$ be a metric on a three-dimensional Riemannian manifold $S$. Given a covariant $k$-tensor field $T$, we first take the symmetric part $T^S$ using
\begin{equation}\label{eq:Spart}
    T^S(X_1,\dots, X_k) = \sum_{\sigma\in S_k} T(X_{\sigma(1)},\dots, X_{\sigma(k)}),
\end{equation}
for $X_1, \dots, X_k\in \mathfrak{X}(S)$ and where $S_k$ is the set of permutations on $\{1,\dots, k\}$. Consequently, we take the trace-free part of $T^{S}$. Then we want to find a symmetric covariant $(k-2)$-tensor field $\tilde{T}$ such that $T^S+h\tilde{T}$ is trace-free, where $h\tilde{T}$ is symmetric product of the tensors $h$ and $\tilde{T}$. It turns out that there is an explicit expression for $\tilde{T}$, most easily expressed in coordinates via \cite{blanchetRadiativeGravitationalFields1986}
\begin{equation}\label{eq:STFpart}
    \tilde{T}_{i_3\dots i_k} = \sum_{m=1}^{\left\lfloor \frac{k}{2}\right\rfloor} A^k_m h_{(i_3i_4}\cdots h_{i_{2m-1}i_{2m}} T^S_{i_{2m+1}\dots i_k)j_1\dots j_{2m}} h^{j_1j_2}\cdots h^{j_{2m-1}j_{2m}},
\end{equation}
with
\begin{equation}\label{eq:STFcoef}
    A^k_m =\frac{(-1)^m k! (2k-2m-1)!!}{2^mm! (k-2m)! (2k-1)!!}.
\end{equation}
Here, the round brackets in the expression for $\tilde{T}$ mean that we take the symmetric part.

We end with a few observations that make life easier when taking symmetric trace-free parts:
\begin{itemize}
    \item For a covariant tensor field $T$, we have
    \begin{equation}\label{eq:tensorprodmetricSTF}
      (T\otimes h)^{STF}=0.
    \end{equation}
    This follows from the bare construction of $\tilde{T}$ above.
    \item For covariant tensor fields $R$ and $T$, we have
    \begin{equation}\label{eq:tensorprodSTF}
      \left(R\otimes T^{STF}\right)^{STF} = \left(R^{STF}\otimes T\right)^{STF} = \left(R\otimes T\right)^{STF}.
    \end{equation}
    This follows from the construction of $T^{STF}$ and \eqref{eq:tensorprodmetricSTF}.
    \item Let $D$ be the total covariant derivative. For a covariant tensor field $T$, we have
    \begin{equation}\label{eq:covdevSTF}
      \left(D\left(T^{STF}\right)\right)^{STF} = (DT)^{STF}.
    \end{equation}
    This follows from the construction of $T^{STF}$, metric-compatibility of the Levi-Civita connection and \eqref{eq:tensorprodmetricSTF}.
\end{itemize}
\end{appendices}

\bibliography{thesis.bib}

\end{document}